\documentclass[11pt]{article}
\usepackage{rpmacros}
\usepackage[dvipsnames,svgnames,x11names,hyperref]{xcolor}
\usepackage{amsthm}
\usepackage{xfrac}
\usepackage{mathpazo}
\usepackage{gitinfo}
\usepackage{enumitem}
\usepackage{fullpage}
\usepackage{thmtools}
\usepackage[T1]{fontenc} 
\usepackage[colorlinks]{hyperref}
\hypersetup{
  linkcolor=[rgb]{0,0,0.4},
  citecolor=[rgb]{0, 0.4, 0},
  urlcolor=[rgb]{0.6, 0, 0}
}
\usepackage{lipsum}
\usepackage{mathrsfs,bbm}
\usepackage{bm}
\usepackage{todonotes}
\usepackage{tikz}
\usetikzlibrary{decorations.pathreplacing,backgrounds}
\usepackage{multirow}
\usepackage{setspace}

\usepackage{microtype}

\usepackage[linesnumbered,vlined,ruled]{algorithm2e}

\usepackage{amsthm}
\usepackage{thmtools,thm-restate}
\usepackage[nameinlink,capitalise]{cleveref}

\numberwithin{equation}{section}
\declaretheoremstyle[bodyfont=\it,qed=\qedsymbol]{noproofstyle}

\declaretheorem[numberlike=equation]{observation}
\declaretheorem[numberlike=equation,style=noproofstyle,name=Observation]{observationwp}
\declaretheorem[name=Observation,numbered=no]{observation*}

\declaretheorem[numberlike=equation]{theorem}

\declaretheorem[name=Theorem,numbered=no]{theorem*}

\declaretheorem[numberlike=equation]{lemma}
\declaretheorem[name=Lemma,numbered=no]{lemma*}
\declaretheorem[numberlike=equation,style=noproofstyle,name=Lemma]{lemmawp}

\declaretheorem[numberlike=equation]{corollary}
\declaretheorem[name=Corollary,numbered=no]{corollary*}

\declaretheorem[name=Proposition,numbered=no]{proposition*}

\declaretheorem[name=Claim,numbered=no]{claim*}

\declaretheorem[name=Conjecture,numbered=no]{conjecture*}

\declaretheorem[name=Question,numbered=no]{question*}

\declaretheoremstyle[bodyfont=\it,qed=$\lozenge$]{defstyle} 

\declaretheorem[numberlike=equation,style=defstyle]{definition}
\declaretheorem[unnumbered,name=Definition,style=defstyle]{definition*}

\declaretheorem[unnumbered,name=Example,style=defstyle]{example*}

\declaretheorem[unnumbered,name=Notation=defstyle]{notation*}

\declaretheorem[unnumbered,name=Construction,style=defstyle]{construction*}

\declaretheorem[unnumbered,name=Remark,style=defstyle]{remark*}

\renewcommand{\phi}{\varphi}
\renewcommand{\epsilon}{\varepsilon}

\newcommand{\size}{\operatorname{size}}

\newcommand{\SP}{\Sigma\Pi}
\newcommand{\SPsize}[1]{(\SP)^k\operatorname{-size}}
\newcommand{\LowDeg}{\operatorname{Deg}}

\usepackage{nth}
\usepackage{intcalc}
\usepackage{etoolbox}
\usepackage{xstring}

\usepackage{ifpdf}
\ifpdf
\else
\usepackage[quadpoints=false]{hypdvips}
\fi

\newcommand{\ECCCurl}[2]{http://eccc.hpi-web.de/report/\ifnumcomp{#1}{>}{93}{19}{20}#1/#2/}

\newcommand{\shortECCC}[2]{\texttt{\href{http://eccc.hpi-web.de/report/\ifnumcomp{#1}{>}{93}{19}{20}#1/#2/}{eccc:TR#1-#2}}}

\newcommand{\parseECCC}[1]{
\StrSubstitute{#1}{TR}{}[\tmpstring]%
\IfSubStr{\tmpstring}{/}{ 
\StrBefore{\tmpstring}{/}[\ecccyear]%
\StrBehind{\tmpstring}{/}[\ecccreport]%
}{
\StrBefore{\tmpstring}{-}[\ecccyear]%
\StrBehind{\tmpstring}{-}[\ecccreport]%
}%
\shortECCC{\ecccyear}{\ecccreport}}

\newcommand{\homog}{\operatorname{Hom}}
\newcommand*\samethanks[1][\value{footnote}]{\footnotemark[#1]}

\newif\ifblind

\newif\ifdraft
\draftfalse
\ifdraft
\newcommand{\RPnote}[1]{\textcolor{BrickRed}{\guillemotleft RP: #1 \guillemotright}}
\newcommand{\MKnote}[1]{\textcolor{Orange}{\guillemotleft MK: #1 \guillemotright}}
\newcommand{\VRnote}[1]{\textcolor{Blue}{\guillemotleft VR: #1 \guillemotright}}
\newcommand{\gitinfonotecolour}{Gray}
\newcommand{\easteregg}{}
\else
\newcommand{\RPnote}[1]{}
\newcommand{\MKnote}[1]{}
\newcommand{\VRnote}[1]{}
\newcommand{\gitinfonotecolour}{white}
\newcommand{\easteregg}{}
\fi

\newcommand{\ignore}[1]{}
\newcommand{\gitinfonote}{git info:~\gitAbbrevHash\;,\;(\gitAuthorIsoDate)\; \;\gitVtag}

\newcommand{\Res}[3]{\ensuremath{\operatorname{Res}_{#1}(#2,#3)}} 
\newcommand{\uglymodel}[2]{\ensuremath{\Sigma \inparen{(\SP)^{(#1)} \cdot (\LowDeg_{#2})^\ast}}}
\renewcommand{\vec}[1]{\ensuremath{\bm{#1}}}
\newcommand{\Ideal}{\ensuremath{\mathcal{I}}}
\newcommand{\Ring}{\ensuremath{\mathcal{R}}}
\newcommand{\bit}{\ensuremath{\operatorname{bit}}} 

\title{Deterministic Algorithms\\for Low Degree Factors\\of Constant Depth Circuits}

\ifblind
\else
\author{
    {Mrinal Kumar \thanks{Tata Institute of Fundamental Research, Mumbai, India. Email: \texttt{\{mrinal, varun.ramanathan, ramprasad\}@tifr.res.in}.  Research supported by the Department of Atomic Energy, Government of India, under project 12-R{\&}D-TFR-5.01-0500.}}
    \and
    {Varun Ramanathan{\samethanks[1]}}
    \and
    {Ramprasad Saptharishi{\samethanks[1]}}
}
\fi
\date{}

\begin{document}

\maketitle

\begin{abstract}
For every constant $d$, we design a subexponential time deterministic algorithm that takes as input a multivariate polynomial $f$ given as a constant depth algebraic circuit over the field of rational numbers, and outputs \textit{all} irreducible factors of $f$ of degree at most $d$ together with their respective multiplicities.  Moreover, if $f$ is a sparse polynomial, then the algorithm runs in quasipolynomial time.

Our results are based on a more fine-grained connection between polynomial identity testing (PIT) and polynomial factorization in the context of constant degree factors and rely on a clean connection between divisibility testing of polynomials and PIT due to Forbes \cite{Forbes15} and on subexponential time deterministic PIT algorithms for constant depth algebraic circuits from the recent work of Limaye, Srinivasan and Tavenas \cite{LST21}.
\end{abstract}

\section{Introduction}
A long line of research (cf. \cite{G83, K85a, Kaltofen92, Kaltofen03}) on the question of designing efficient algorithms for multivariate polynomial factorization concluded with the influential works of Kaltofen \cite{K89} and Kaltofen \& Trager \cite{KT88} which gave efficient  randomized algorithms for this problem in the whitebox and blackbox settings respectively.\footnote{Throughout this paper, we use \textit{efficient} to mean an algorithm whose time complexity is polynomially bounded in the size, bit-complexity and the degree of the input algebraic circuit.} These results and the technical insights discovered in the course of their proofs have since found numerous direct and indirect applications in various areas of complexity theory. This includes applications to the construction of pseudorandom generators for low degree polynomials \cite{bogdanov05},  algebraic algorithms \cite{KY08}, hardness-randomness tradeoffs in algebraic complexity \cite{DSY09, GKSS19}, algebraic property testing \cite{PS94, AS03, BCIKS20}, error correcting codes \cite{BHKS20}, deterministic polynomial identity tests for constant depth circuits \cite{LST21, ChouKS18} among others. 

Given the fundamental nature of the problem and its many applications, the question of designing efficient deterministic algorithms for multivariate polynomial factorization is of great interest and importance. Shpilka \& Volkovich \cite{SV10} observed that this question is at least as hard as PIT in the sense that a deterministic factoring algorithm (in fact, an algorithm to check irreducibility suffices for this) for polynomials given by algebraic circuits implies a deterministic algorithm for PIT for algebraic circuits, a long standing open problem in computer science. In a later work, Kopparty, Saraf \& Shpilka \cite{KSS15} showed a connection in the other direction as well. They showed that an efficient deterministic algorithm for PIT for algebraic circuits implies an efficient deterministic algorithm for polynomial factorization for algebraic circuits. Thus, the questions are essentially equivalent to each other. 

An intriguing aspect of the aforementioned equivalence is that while deterministic algorithms for factoring any rich enough class of circuits (for instance, constant depth circuits) lead to deterministic PIT for the same class (see Observation 1 in \cite{SV10} for a precise statement), the connection in the other direction due to Kopparty, Saraf \& Shpilka \cite{KSS15} does not appear to be so fine-grained. In particular,  even if we only wish to factor an otherwise simple class of polynomials, e.g. sparse polynomials (polynomials with a small number of non-zero monomials), the PIT required as per the proof in \cite{KSS15} seems to be for significantly more powerful models of algebraic computation like algebraic branching programs. 

As a consequence, while there has been steady progress on the state of the art of deterministic PIT algorithms in  recent years for various interesting sub-classes of algebraic circuits like sparse polynomials \cite{KS01}, depth-$3$ circuits with constant top fan-in \cite{SS09, SS10, KS09a}, read-once algebraic branching programs \cite{FS13, FSS14, F14, GKST15, GKS16} and constant depth circuits \cite{LST21}, this progress hasn't translated to progress on the question of deterministic factoring algorithms for these circuit classes. In particular, deterministic factorization algorithms have remained elusive even for seemingly simple classes of polynomials like sparse polynomials where the corresponding PIT problem is very well understood. There are only a handful of results that make progress towards this and related problems to the best of our knowledge. Shpilka \& Volkovich \cite{SV10} showed a close connection between the problems of polynomial identity testing and that of decomposing a polynomial given by a circuit into variable disjoint factors and build on these ideas to give an efficient deterministic algorithm for factoring sparse multilinear polynomials. In subsequent works, Volkovich \cite{Volkovich15, Volkovich17} gave an efficient deterministic algorithm to factor sparse polynomials that split into multilinear factors and  sparse polynomials with individual degree at most $2$. More recently, a work of Bhargava, Saraf and Volkovich \cite{BSV18}  gives a quasipolynomial time deterministic algorithm for factoring sparse polynomials with small individual degree based on some beautiful geometric insights.  

In general, when the individual degree of a sparse polynomial is not small, no non-trivial deterministic factoring algorithms appear to be known, even when we have the flexibility of describing the output as algebraic circuits. As Forbes \& Shpilka note in their recent survey \cite{Forbes-Shpilka-survey} on polynomial factorization, we do not even have structural guarantees on the complexity of factors of sparse polynomials even for seemingly coarse measures of complexity like formula complexity. In fact, questions that might be potentially easier than factorization like checking if a given sparse polynomial is a product of constant degree polynomials or checking if a given sparse polynomial is irreducible are not known to have non-trivial deterministic algorithms. Perhaps a little surprisingly, till a recent work of Forbes \cite{Forbes15}, we did not even have a non-trivial deterministic algorithm for checking if a given sparse polynomial is divisible by a given constant degree polynomial! Forbes gave a quasipolynomial time deterministic algorithm for this problem by reducing this question to a very structured instance of PIT for depth-$4$ algebraic circuits and then giving a quasipolynomial time deterministic algorithm for these resulting PIT instances. 

This work is motivated by some of these problems, most notably by the question of designing efficient deterministic algorithms for factoring sparse polynomials. While we do not manage to solve this problem in this generality, we make modest progress towards this: we design a deterministic quasipolynomial time  algorithm that outputs all the low degree factors of a sparse polynomial. More generally,  we show that constant degree factors of a polynomial given by a constant depth circuit can be computed deterministically in subexponential time.

\subsection{Our Results} 

\begin{theorem}[Low degree factors of constant depth circuits]\label{thm: main low depth ckts}
  Let $\Q$ be the field of rational numbers and $\epsilon > 0$, $d, k \in \N$ be arbitrary constants. 

Then, there is a deterministic algorithm  that takes as input an algebraic circuit $C$ of size $s$, bit-complexity $t$, degree $D$ and depth $k$ and outputs all the irreducible factors of $C$ of degree at most $d$, along with their respective multiplicities in time $(sDt)^{O((sDt)^{\epsilon})}$. 

\end{theorem}
We note that the bit-complexity of an algebraic circuit/formula is a measure of the bit-complexities of the rational numbers appearing in the circuit. See \autoref{defn:circuit-size-bit-complexity} for a formal definition. 

When the input polynomial is sparse, i.e. has a small depth-$2$ circuit, then the time complexity of the algorithm in \cref{thm: main low depth ckts} can be improved to be quasipolynomially bounded in the input size. This gives us the following theorem. 
\begin{theorem}[Low degree factors of sparse polynomials]\label{thm: main sparse}
Let $d \in \N$ be an arbitrary constant. 

Then, there is a deterministic algorithm that takes as input a polynomial $f \in \Q[\vecx]$ of sparsity $s$, bit-complexity $t$, degree $D$, and outputs all the irreducible factors of $f$ of degree at most $d$, along with their respective multiplicities in time $(sDt)^{\poly(\log sDt)}$.  

\end{theorem}
These results immediately yield an algorithm (with comparable time complexity) to check if the polynomial computed by a given low depth circuit is a product of polynomials of degree at most $d$. More concretely, we have the following corollary that follows by comparison of degrees of the input polynomial and the low-degree factors (with multiplicities) listed by the algorithms in the above theorems. 
\begin{corollary}\label{cor: product of low degree factors}
Let $\Q$ be the field of rational numbers and $\epsilon > 0$, $d, k \in \N$ be arbitrary constants. 

Then, there is a deterministic algorithm that takes as input an algebraic circuit $C$ of size $s$  , bit-complexity $t$, degree $D$ and depth $k$ and decides if $C$ is a product of irreducibles of degree at most $d$ in time $(sDt)^{O((sDt)^{\epsilon})}$. 

Moreover, when $f$ is a sparse polynomial with sparsity $s$, then the algorithm runs in $(sDt)^{\poly(\log sDt)}$ time. 
\end{corollary}
Note that in the constant depth regime, circuits and formulas are equivalent upto a polynomial blow-up in size. Thus we will use the terms circuits and formulas interchangeably without any loss in our final bounds, and most of our presentation will be for formulas.

\subsection*{Field dependence of our results}
We end this section with a remark about the field dependence of our results. 
The field dependence in our results stems from two reasons. We need an efficient deterministic algorithm for factorization of univariate polynomials over the underlying field $\F$. In addition to this, our proofs also need non-trivial deterministic algorithms for polynomial identity testing (PIT) for constant depth circuits (or very special depth-$4$ circuits for \cref{thm: main sparse}) over the underlying field. 

The field of rational numbers satisfies both these requirements: a classical algorithm of Lenstra, Lenstra and {\Lovasz }\cite{LLL82} solves the problem of deterministic univariate factorization efficiently over $\Q$ and a recent work of Limaye, Srinivasan and Tavenas \cite{LST21}  gives a subexponential time deterministic algorithm for PIT for constant depth circuits over $\Q$. For \cref{thm: main sparse}, the relevant PIT is for special depth-$4$ circuits and was given in a work of Forbes \cite{Forbes15}. In fact, Forbes' result holds even over finite fields.

We restrict our attention to just the field of rational numbers in the presentation although our results work over any large characteristic field that supports the above requirements. 

\subsection{Proof Overview}\label{sec: proof overview}

We now give an overview of some of the main ideas in our proofs. In a nutshell, our proofs are based on relatively simple structural observations on top of the existing factoring algorithms. The key is to understand  the structure of circuits for which we need a PIT algorithm at every step a little better, and when looking for low degree factors, we observe that these PIT instances are relatively simple and their circuit complexity is comparable to the circuit complexity of the input polynomials themselves. We also crucially use the divisibility testing idea of Forbes \cite{Forbes15} in our algorithm at two stages; this helps us handle factors of large multiplicities and also lets us obtain true factors from the output of Hensel Lifting step of the factorization algorithms. This idea again helps in reducing the complexity of the PIT instance we face in these steps, and in particular, we completely avoid the linear systems solving step in a typical factorization algorithm that naively (e.g. see \cite{KSS15}) seems to require PIT for algebraic branching programs. Once the PIT instances are shown to be relatively simple, we invoke the PIT algorithms of Forbes \cite{Forbes15} and Limaye, Srinivasan \& Tavenas \cite{LST21} to solve these deterministically.

\paragraph{Typical steps in a polynomial factorisation algorithm:} Most factorisation algorithms (and ours, modulo minor deviations) follow this template:
\begin{enumerate}\itemsep 0pt
\item \textbf{Making $f$ monic:} Apply a suitable transformation of the form $x_i \mapsto x_i + \alpha_i y$ to ensure that $f$ is monic in $y$. We may now assume that $f \in \Q[\vecx,y]$. 
\item \textbf{Preparing for Hensel lift:} Ensure that $f(\vecx,y)$ is square-free, and further that $f(\veczero,y)$ is also square-free. 
\item \textbf{Univariate factorisation:} Factorise the univariate polynomial $f(\veczero,y)$ as a product $g_0(y)\cdot h_0(y)$ where $\gcd(g_0,h_0) = 1$. This can be intepreted as a factorisation $f(\vecx,y) = g_0 \cdot h_0 \bmod{\mathcal{I}}$ where $\mathcal{I} = \inangle{\vecx}$. 
\item \textbf{Hensel lifting:} Compute an iterated lift to obtain $f = g_\ell \cdot h_\ell \bmod{\mathcal{I}^{2^\ell}}$ for a suitably large $\ell$. 
\item \textbf{Reconstruction:} From $g_\ell$, obtain an honest-to-god factor $g$ of $f$ (unless $f$ is irreducible). 
\end{enumerate}

The first two steps typically involve the use of randomness for suitable polynomial identity tests. In the first step, we would like $\vecalpha$ to be a point that keeps the highest degree homogeneous component of $f$ non-zero, and the second step is handled by translating $f$ by a point $\vecdelta$ that keeps the ``discriminant'' of $f$ non-zero. The Hensel lift is a deterministic subroutine that eventually yields small circuits for the lifted factors and the reconstruction step typically involves solving a linear system. It is mostly due to the ``discriminant'' that we do not have efficient deterministic factorisation algorithm even for constant-depth circuits as the best upper bound for the discriminant we have is an algebraic branching program and we do not have efficient hitting sets for them. \textcolor{white}{(Yet!)}

\medskip

For our case, it is instructive to focus on a specific factor $g$ of $f$ and understand what would be required to make the above template yield this factor. The first observation is that the base case of Hensel Lifting does not require $f$ to be square-free but rather that the factor $g$ we intend to reconstruct satisfies  $g | f$ and $g^2 \nmid f$. For now, let us assume this and also that $f$ (and hence $g$ and $h = f/g$ also) is monic in $y$. We have that $\gcd(g,h) = 1$ but for the Hensel lift, we also need to find a $\vecdelta$ that ensures that $\gcd(g_0,h_0) = 1$ where $g_0 = g(\vecdelta,y)$ and $h_0 = h(\vecdelta,y)$. The set of ``good'' $\vecdelta$'s is precisely the points that do not make the resultant $\Res{y}{g}{h}$ zero and thus we want to understand the circuit complexity of this resultant. 

The resultant $\Res{y}{g}{h}$ is the determinant of a matrix of dimension $\deg_y(g) + \deg_y(h)$ and its entries are coefficients of $g, h$ when viewed as univariates in $y$. However, we are only given that $f$ is computable by a constant-depth formula and we do not have any good bound on the complexity of $h$. We circumvent this by working with a \emph{pseudo-quotient} (introduced by Forbes~\cite{Forbes15} in the context of divisibility testing) $\tilde{h}$ of $f$ and $g$; we work with $\Res{y}{g}{\tilde{h}}$ and show that it is also computable by constant-depth circuits of not-too-large size. Fortunately, the result of Limaye, Srinivasan and Tavenas~\cite{LST21} yields sub-exponential sized hitting sets for constant depth formulas and that enables us to avoid the use of randomness to prepare for the Hensel Lifting step. 

We can then factorise the univariate polynomial $f(\vecdelta,y)$ and attempt all possible factors $g_0$ of degree at most $d$ to begin the lifting process from $g_0 \cdot h_0$ (where $h_0 = f(\vecdelta,y)/g_0$). After an appropriately large lift, we have small circuits (of possibly unbounded depth) computing $g_\ell$ and $h_\ell$ such that $\tilde{f} = f(\vecx + \vecdelta, y) = g_\ell \cdot h_\ell \bmod{\mathcal{I}^{2^\ell}}$. If $g_\ell$ is guaranteed to be monic, and the initial choice of $g_0$ was indeed $g(\vecdelta,y)$, the uniqueness of Hensel lifting would ensure that $g_\ell$ is indeed equal to $g$ (after truncating higher order terms). We can then use standard interpolation to obtain $g_\ell$ explicitly written as a sum of monomials. Finally, to ensure that $g_\ell$ is indeed a legitimate factor of $\tilde{f}$, we perform divisibility testing to check if $g_\ell \mid \tilde{f}$. 

\paragraph{Handling factors of large multiplicity:} The above overview is all we need to obtain any factor $g$ of degree $O(1)$ that divides $f$ with $g^2 \nmid f$. In order to handle factors with higher ``factor-multiplicity'', we use a simple observation that $g^{a-1} \mid f$ but $g^a \nmid f$ if and only if $g$ divides $f, \partial_y f, \ldots, \partial_{y^{a-1}}f$ but not $\partial_{y^{a}}f$. We run our algorithm for each of the partial derivatives to collect the list of candidate factors, and eventually prune them via appropriate divisibility tests. 

\paragraph{The specific case of $\SP$-formulas (or sparse polynomials):} The above sketch yields a sub-exponential time algorithm for obtaining $O(1)$-degree factors of constant depth formulas. However, with some additional care, we obtain a quasipolynomial time algorithm in the case when $f$ is a sparse polynomial. The key observation for this is that we do not really need $f$ to be made monic for the above approach, but we only need $g$ to be monic to exploit the uniqueness of Hensel lifts. Since $g$ is a polynomial of degree at most $d = O(1)$, we can find a \emph{low Hamming weight} vector $\vecalpha$ such that $g(\vecx + y \vecalpha)$ is monic in $y$. This allows us to control the sparsity increase of $f$ in the process and we show that the relevant resultant is a polynomial of the form 
\[
  \sum_i \text{monomial}_i \cdot (\text{$O(1)$-degree})^{e_i}.
\]
Forbes~\cite{Forbes15} shows that there are quasipolynomial size hitting sets for such expressions and we use this instead of the more general hitting set of Limaye, Srinivasan and Tavenas~\cite{LST21}. 

\subsection*{Organization of the paper}
The rest of the paper is organized as follows. 

In the next section, we start with a discussion of some of the preliminaries and known results from algebraic complexity and previous works on polynomial factorization that we use for the design and analysis of our algorithms. In \autoref{sec:computing-mult-one-factors}, we describe and analyze the algorithm for computing low degree factors of multiplicity one of a given constant depth formula. In \autoref{section: arbitrary multiplicity}, we build upon this algorithm to compute arbitrary constant degree factors and complete the proofs of \autoref{thm: main low depth ckts} and \autoref{thm: main sparse}. Finally, we conclude with some open problems in \autoref{section: open problems}.

\section{Notation and preliminaries}\label{section:notations and prelims}

This section consists of all the necessary building blocks to describe and analyse (in \cref{sec:computing-mult-one-factors}) the main algorithm. 

\paragraph{Fair warning:} A large part of this (slightly lengthy) section is standard techniques in algebraic complexity that are relevant to this specific context, and is intended to keep the main analysis as self-contained as possible. A reader with some familiarity with standard algorithmic and structural results in algebraic complexity might be in a position to directly proceed to \cref{sec:computing-mult-one-factors} and revisit this section for relevant results as required. 

\subsubsection*{Notation}

\begin{enumerate}\itemsep0pt
\item Throughout this paper, we work over the field $\Q$ of rational numbers. For some of the statements that are used more generally, we use $\F$ to denote an underlying field.  

\item We use boldface lower case letters like $\vecx, \vecy, \veca$ to denote tuples, e.g. $\vecx = (x_1, x_2, \ldots, x_n)$. The arity of the tuple is either stated or will be clear from the context. 
\item For a polynomial $f$ and a non-negative integer $k$, $\homog_k[f]$ denotes the homogeneous component of $f$ of degree \emph{equal} to $k$. $\homog_{\leq k}[f]$ denotes the sum of homogeneous components of $f$ of degree at most $k$, i.e.,
\[
\homog_{\leq k}[f] := \sum_{i = 0}^k \homog_i[f].
\]
\item The \emph{sparsity} of a polynomial $f$ is the number of monomials with a non-zero coefficient in $f$.

\item For a parameter $k\in \Z_{\geq 0}$, we will use $(\SP)^{(k)}$ to refer to product-depth $k$ circuits\footnote{We emphasize that this notation does \emph{not} refer to the $k^{\text{th}}$ power of a polynomial computed by a $\SP$ circuit.} with the root gate being $+$ and the deepest layer of gates being $\times$. Since any constant depth algebraic circuit of depth $k$ and size $s$ can be converted to a formula of depth $k$ and size $s^{k+1}$ i.e. $\poly(s)$, we will use the terms circuits and formulas interchangeably, without any loss in the final bounds we prove.

\item Let $f$ and $g$ be multivariate polynomials such that $g \mid f$. Then, the \emph{multiplicity} or \emph{factor multiplicity} of $g$ in $f$ is defined to be the greatest integer $a$ such that $g^a$ divides $f$.
\end{enumerate}

\subsection{Circuit/formula bit-complexity}

\begin{definition}[Bit-complexity of a circuit/formula]
  \label{defn:circuit-size-bit-complexity}
  The \emph{bit-complexity} of a circuit/formula $C$, denoted by $\bit(C)$, is defined as the sum of $\size(C)$ and the bit-complexities of all the scalars\footnote{For a rational number $r = p/q$, its bit-complexity $\bit(r)$ is defined as $\log(\max(\abs{p}, \abs{q}))$} present on edges or leaves. By default, any edge that does not have a scalar on it will be assigned the scalar 1. 
\end{definition}

\begin{lemma}[Bit-complexity of evaluations of formulas]
  \label{lem:evaluation-bit-complexity}
  Let $C$ be a formula of bit-complexity $s$ computing a polynomial $f(\vecx)$. If $\veca \in \Q^n$ with each entry of $\veca$ having bit-complexity $b$, then the bit-complexity of $f(\veca)$ is at most $s \cdot b$.
\end{lemma}

\noindent (Proof deferred to \cref{sec:appendix:deferred_proofs})

\subsection{Relevant subclasses of algebraic circuits} 

We briefly define subclasses of algebraic circuits that we would use often in this paper.

\begin{definition}[Power of low-degree polynomials]
  For a parameter $d \in \Z_{\geq 0}$, let $\LowDeg_d$ refer to the class of polynomials of degree at most $d$. We use $(\LowDeg_d)^\ast$ to denote the class of polynomials that are powers of polynomials of degree at most $d$. 
\end{definition}

\begin{definition}[$\uglymodel{k}{d}$-formulas]
  \label{defn:ugly-model}
  We will use $\uglymodel{k}{d}$ to denote the subclass of algebraic formulas that compute expressions of the form
  \[
    \sum_i f_i \cdot g_i^{e_i}
  \]
  where each $f_i$ is a $(\SP)^{(k)}$ formula and each $g_i$ is a polynomial of degree at most $d$ and $e_i$'s are arbitrary positive integers. The size and bit-complexity of the above expression is defined as its size and bit-complexity when viewed as a general algebraic formula. 
\end{definition}

\begin{observation}
  \label{obs:ugly-model-sums-products}
  Let $\mathcal{C}$ be the class of $\uglymodel{k}{d}$ formulas for fixed parameters $k$ and $d$. 
  Suppose $P_1, \ldots, P_t$ are polynomials computed by $\uglymodel{k}{d}$ formulas of size $s$ and bit-complexity $b$ each. Then, 
  \begin{itemize}
    \item $\sum_i P_i$ is computable by an $\uglymodel{k}{d}$ formula of size at most $t \cdot s$ and bit-complexity at most $O(t \cdot b)$.
    \item $\prod_i P_i$ is computable by an $\uglymodel{k}{d}$ formula of size at most $s^{O(t)}$ and bit-complexity at most $b^{O(t)}$. 
  \end{itemize}
\end{observation}

\noindent (Proof deferred to \cref{sec:appendix:deferred_proofs}.)

\subsection{Standard preliminaries using interpolation}

\begin{lemmawp}[Univariate interpolation (Lemma 5.3 \cite{S15})]\label{lem:interpolation-univariate}
  Let $f(x) = f_0 + f_1 x + \cdots + f_d x^d$ be a univariate polynomial of degree at most $d$. Then, for any $0 \leq r \leq d$ and there are\footnote{In fact, for any choice of distinct $\alpha_0, \ldots, \alpha_d$, there are appropriate $\beta_{r0}, \ldots, \beta_{rd}$ satisfying the equation. If the $\alpha_i$'s are chosen to have small bit-complexity, we can obtain a $\poly(d)$ bound on the bit-complexity of the associated $\beta_{ri}$'s.} field constants $\alpha_0, \ldots, \alpha_d$ and $\beta_{r0},\ldots, \beta_{rd}$ such that 
  \[
    f_r = \beta_{r0} f(\alpha_0) + \cdots + \beta_{rd} f(\alpha_d).
  \]
  Furthermore, the bit-complexity of all field constants is bounded by $\poly(d)$. 
\end{lemmawp}

\begin{lemma}[Computing homogeneous components (Lemma 5.4 \cite{S15})]
  \label{lem:computing-hom-components}
  Let $f \in \Q[\vecx]$ be an $n$-variate degree $d$ polynomial. Then, for an $0 \leq i \leq d$, there are field constants $\alpha_0,\ldots, \alpha_{d}$ and  $\beta_{i0}, \beta_{id}$ of bit-complexity $\poly(d)$ such that 
  \[
    \homog_i(f) = \beta_{i0} f(\alpha_0 \cdot \vecx) + \cdots + \beta_{id} f(\alpha_d \cdot \vecx).
  \]
  In particular for $\mathcal{C} = (\SP)^{(k)}$ or $\Sigma \inparen{(\SP)^{(k)} \cdot (\LowDeg_d)^\ast}$, if $f$ is computable by $\mathcal{C}$-formulas of size / bit-complexity at most $s$ then $\homog_i(f)$ is computable by $\mathcal{C}$-formulas of size / bit-complexity at most $\poly(s,d)$.
\end{lemma}

\begin{lemma}[Computing partial derivatives in one variable]
  \label{lem: partial derivatives with respect to one variable}
  Let $f \in \Q[\vecx]$ be an $n$-variate degree $d$ polynomial. Then, for an $0 \leq r \leq d$, there are field elements $\alpha_i$'s and $\beta_{ij}$'s in $\Q$ of bit-complexity $\poly(d)$ such that 
  \[
    \frac{\partial^r f}{\partial x_1^r} = \sum_{i=0}^{d} x_1^i \cdot \inparen{\beta_{i0} f(\alpha_0, x_2,\ldots, x_n) + \cdots + \beta_{id} f(\alpha_d, x_2,\ldots, x_n)}
  \]
  In particular for $\mathcal{C} = (\SP)^{(k)}$ or $\Sigma \inparen{(\SP)^{(k)} \cdot (\LowDeg_d)^\ast}$, if $f$ is computable by $\mathcal{C}$-formulas of size / bit-complexity at most $s$ then $\frac{\partial^r f}{\partial x_1^r}$ is computable by $\mathcal{C}$-formulas / bit-complexity of size at most $O(s \cdot d^3)$.
\end{lemma}
\begin{proof}
  We may consider the polynomial $f$ as a univariate in $x_1$, and extract each coefficient of $x_1^i$ using \cref{lem:interpolation-univariate} and recombine them to get the appropriate partial derivative. That justifies the claimed expression. 

  As for the size, note that if $\mathcal{C}$ is $(\SP)^{(k)}$ or $\Sigma \inparen{(\SP)^{(k)} \cdot (\LowDeg_d)^\ast}$, multiplying a size $s$ formula by $x_1^i$, by using distributivity of the top addition gate, results in a $\mathcal{C}$-formula of size at most $s \cdot d$. Thus, the overall size of the above expression for the partial derivative is at most $O(s \cdot d^3)$. 
\end{proof}

We will be making use of the following identity, which can be proved via appropriate interpolation or by the inclusion-exclusion principle (along the lines of Lemma 2.2 \cite{Shpilka02}).

\begin{lemma}[Fischer's identity \cite{F94, E69, Shpilka02}]
  \label{lem:fischers-trick}
  If $\F$ is a field of characteristic zero or larger than $D$, then for any positive integers $e_1,\ldots, e_n$ with $\sum e_i = D$ and for $r \leq \prod_{i=1}^{n}\inparen{e_i+1}$, there are homogeneous linear forms $L_1,\ldots, L_r$ and field constants $\alpha_1, \ldots, \alpha_r$ of bit-complexity $\poly(d,n)$ such that 
  \[
  x_1^{e_1} \cdots x_n^{e_n} = \sum_{i=1}^r \alpha_i L_i^D.
  \]
\end{lemma}

\subsection{Polynomial identity testing}

\begin{lemma}[Polynomial Identity Lemma \cite{O22,DL78,S80,Z79}]\label{lem: SZ lemma}
Let $f \in \Q[\vecx]$ be a non-zero $n$ variate polynomial of degree at most $d$. Then, for every set $S \subseteq \Q$, the number of zeroes of $f$ in the set $S^n = S \times S \times \cdots \times S$ is at most $d|S|^{n-1}$. 
\end{lemma}

\begin{definition}[Low Hamming weight set]
  \label{defn:low-hamming-weight-hitting-set}
  Let $n \geq d \geq 0$ be integer parameters. Fix a set $T_d \subseteq \Q$ of size $(d+1)$. The set $\mathcal{H}(d,n)$ is defined as
  \[
    \mathcal{H}(d, n) = \setdef{(a_1,\ldots, a_n)}{S \in \binom{[n]}{\leq d} \;,\; a_i \in T_d \text{ for all } i\in S \text{ and } a_j = 0 \text{ for all } j\notin S}.
  \]
  The size of the above set is at most $\binom{n}{\leq d} \cdot (d+1)^d = n^{O(d)}$. Furthermore, choosing $T_d$ to consist of elements of $\Q$ of bit-complexity $\poly(d)$, the bit-complexity of the set $\mathcal{H}(d,n)$ is bounded by $n^{O(d)}$ as well. 
\end{definition}

The following lemma is an easy consequence of \cref{lem: SZ lemma} and will be crucial for parts of our proof. We also include a short proof sketch. 
\begin{lemma}[Hitting set for low degree polynomials]\label{lem: non-zeros of low degree polynomials}
Let $f \in \Q[\vecx]$ be a non-zero $n$ variate polynomial of degree at most $d$. 
Then, there exists a vector $\veca \in \mathcal{H}(d, n) \subseteq \Q^n$ such that $f(\veca) \neq 0$.
\end{lemma}

\noindent (Proof deferred to \cref{sec:appendix:deferred_proofs}.)

\begin{theorem}[PIT for constant depth formulas (modification of Corollary 6 \cite{LST21})]\label{thm: PIT from LST}
Let $\epsilon > 0$ be a real number and $\F$ be a field of characteristic 0. Let $C$ be an algebraic formula of size and bit-complexity $s \leq \operatorname{poly}(n)$, depth $k = o(\log\log\log n)$ computing a polynomial on $n$ variables, then there is a deterministic algorithm that can check whether the polynomial computed by $C$ is identically zero or not in time $(s^{O(k)}\cdot n)^{O_{\epsilon}((sD)^\epsilon)}$.
\end{theorem}

The original statement of Corollary 6 in \cite{LST21} deals specifically with circuits of size $s = \poly(n)$. The above statement can be readily inferred from their proof. 

\begin{theorem}[PIT for $\uglymodel{1}{d}$ (Corollary 6.7, \cite{Forbes15})]\label{thm: forbes quasipoly pit for resultant}
  Let $t \geq 1$. Then, the class $\mathcal{C} = \uglymodel{1}{d}$ that computes polynomials of the form $\sum_{i=1}^s f_i \cdot g_i^{d_i}$ with each $f_i$ being $s$-sparse and each $\deg(g_i) \leq d$ has a $\poly(n,s, d\log s)$-explicit hitting set of size $\poly(s)^{O(d \log s)}$. 
\end{theorem}

We will also crucially use the following lemma that gives an algorithm to obtain the coefficient vector of a polynomial from an algebraic formula computing it. In our setting, we invoke this algorithm only for low degree polynomials, and in that case, we can tolerate the runtime of this algorithm within our budget. 
\begin{lemma}[Interpolating a low degree multivariate polynomial]\label{lem: multivariate interpolation}
There is a deterministic algorithm that, when given a parameter $d$ and an $n$ variate algebraic formula $C \in \Q[\vecx]$ of size at most $s$, bit-complexity at most $b$ and degree at most $d$, outputs the coefficient vector of the polynomial computed by $C$. 

The algorithm runs in time $\poly(s, b, n^d)$.
\end{lemma}

\noindent (Proof deferred to \cref{sec:appendix:deferred_proofs}.)

\subsection{Deterministic divisibility testing and PIT}\label{section: div test and PIT}

\begin{definition}[Pseudo-quotients]
  \label{defn:pseudo-quotient}
  Let $f,g \in \Q[\vecx]$ be non-zero polynomials with $g(\veczero) = \beta \neq 0$. The \emph{pseudo-quotient of $f$ and $g$} is defined as
  \[
    \homog_{\leq d_f - d_g}\inparen{\inparen{\frac{f(\vecx)}{\beta}} \cdot (1 + \tilde{g} + \tilde{g}^2 + \cdots + \tilde{g}^{d_f - d_g})}
  \]
  where $d_f = \deg(f)$, $d_g = \deg(g)$ and $\tilde{g} = 1 - \frac{g}{\beta}$.

  More generally, if $\vecalpha \in \Q^n$ is such that $g(\vecalpha) \neq 0$, the \emph{pseudo-quotient of $f$ and $g$ translated by $\vecalpha$} is defined as the pseudo-quotient of $f(\vecx + \vecalpha)$ and $g(\vecx + \vecalpha)$.
\end{definition}

The following lemma immediately follows from the above definition and \cref{lem:computing-hom-components}.

\begin{lemma}[Complexity of pseudo-quotients]
  \label{lem:complexity-pseudoquotient}
  Suppose $k \geq 1$ and $f(\vecx) \in (\SP)^{(k)}$ and $g(\vecx) \in \LowDeg_d$ of sizes at most $s_1, s_2$ respectively, and suppose $g(\veczero) \neq 0$. Then, the pseudo-quotient of $f,g$ is computable by the $\mathcal{C}$-formulas of size at most $\poly(s_1,s_2)$, where $\mathcal{C} = \Sigma \inparen{(\SP)^{(k)} \cdot (\LowDeg_d)^\ast}$.
\end{lemma}

\begin{theorem}[Divisibility testing to PIT \cite{Forbes15}]\label{thm: divisibility testing to pit}
Let $f(\vecx)$ and $g(\vecx)$ be non-zero $n$-variate polynomials over a field $\Q$ such that $g(\mathbf{0}) = \beta \neq 0$. Then, $g$ divides $f$ if and only if the polynomial $R(\vecx)$ defined as  
\[
R(\vecx) := f(\vecx) - g(\vecx) Q(\vecx)
\]
is identically zero, where $Q(\vecx)$ is the pseudo-quotient of $f$ and $g$. 
\end{theorem}

An immediate consequence of this theorem is the following corollary that takes into account the depth of an algebraic formula computing the polynomial $R(\vecx)$ given above, assuming that $f$ and $g$ themselves can be computed by a low depth formula. 
\begin{corollary}[Divisibility testing to PIT for constant depth formulas \cite{Forbes15}]\label{cor: divisibility to pit constant depth}
  Suppose $f(\vecx)$ is a non-zero $n$-variate polynomial computed by a $(\SP)^{(k)}$ formula of size $s$, and suppose $g(\vecx)$ is a polynomial of degree at most $d$ with $g(\veczero) = \beta \neq 0$. Then, we can test if $g$ divides $f$ in time $T(k,d,s')$ where $s' = \poly(s,d)$ and $T(k,d,s)$ is the time required to test polynomial identities of the size $s$ expressions of the form
  \[
    \Sigma \inparen{(\SP)^{(k)} \cdot (\LowDeg_d)^\ast}.
  \]
\end{corollary}

\noindent (Proof deferred to \cref{sec:appendix:deferred_proofs}.)

\begin{theorem}[\cite{Forbes15}]\label{thm: det algo for low degree dividing sparse}
Let $\F$ be any sufficiently large field. Then, there is a deterministic algorithm that takes an input two polynomials $f$ and $g$ and parameters $d, D, n, s$, where $f$ is an $n$-variate polynomial of degree at most $D$ and sparsity $s$; $g$ is an $n$ variate polynomial of degree $d$, and outputs whether $g$ divides $f$ or not in time $\exp(O(d\log^2 snDd))$.
\end{theorem}

\subsection{Resultants}\label{subsection: resultants}

\begin{definition}[The Resultant]\label{defn: resultant}
  Let $\mathcal{R}$ be a commutative ring. Given polynomials $g$ and $h$ in $\Ring[y]$, where:
  \begin{align*}
    g(y) & = g_0 + \cdots + y^d \cdot g_d \\
    h(y) & = h_0 + y \cdot h_1  + \cdots + y^D \cdot h_D
  \end{align*}
  with $g_d$ and $h_D \neq 0$ the \emph{Resultant} of $g$ and $h$, denoted by $\Res{y}{g}{h}$, is the determinant of the $(D + d) \times (D + d)$ Sylvester matrix $\Gamma$ of $g$ and $h$, given by: 
  \[\Gamma = 
    \begin{bmatrix}
        h_0  & h_1    & \dots  &        & h_{D} &        &       \\
            & \ddots & \ddots &        & \ddots & \ddots &       \\
            &        & h_0    & h_1    &        &  \dots & h_{D} \\
      g_{0}  & \dots  &        & g_{d}  &        &        &       \\
            &g_{0}   & \dots  &        & g_{d}  &        &       \\
            &        & \ddots & \ddots &        & \ddots &       \\
            &        &        & g_{0}  & \dots &        &g_{d} 
    \end{bmatrix}
  \]
\end{definition}

\begin{lemma}[Resultant and $\gcd$ (Corollary 6.20 \cite{GG13})]\label{lem: resultant property}
    Let $\Ring$ be a unique factorization domain and $g, h \in \Ring[y]$ be non-zero polynomials. Then: $$\deg_y(\gcd(g,h)) > 0 \iff \Res{y}{g}{h} = 0$$
    where $\gcd(g,h) \in \Ring[y]$ and $\Res{y}{g}{h} \in \Ring$.
\end{lemma}

In this paper, $\Ring$ will be $\Q[\vecx]$ (which is a unique factorization domain), and $\Res{y}{g}{h}$ will denote the resultant of $g, h \in \Q[\vecx][y]$ when considered as polynomials in $\Ring[y]$. We might also occasionally refer to it as the \emph{y-resultant} of $g$ and $h$. For more details about the resultant as well as a proof of the above lemma, we refer the reader to von zur Gathen and Gerhard's book on computer algebra (Chapter 6, \cite{GG13}). We mention a simple observation from the above definition that would be useful for this paper. 

\begin{observationwp}[Resultant under substitutions]
  \label{obs:resultant-under-substitution}
  Suppose $g(\vecx, y) = g_0(\vecx) + g_1(\vecx)y + \cdots g_d(\vecx) y^d$ and $h(\vecx,y) = h_0(\vecx) + h_1(\vecx)y + \cdots + h_D(\vecx) y^D$ with $g_d, h_D \neq 0$. Then, for any $\veca \in \Q^{\abs{\vecx}}$ that ensures $g_d(\veca), h_D(\veca) \neq 0$, we have
  \[
    (\Res{y}{g}{h})(\veca) = \Res{y}{g(\veca,y)}{h(\veca,y)}.\qedhere
  \]
\end{observationwp}

\subsection{Hensel Lifting}
Now we will state the definition of a \emph{lift} and the main lemma for Hensel lifting. 
For more details, one can look up some of the cited papers or the standard references in computational algebra \cite{KSS15, ST20, GG13, Sudan98}.

\begin{definition}[Hensel lifts]\label{defn: lift}
  Let $\mathcal{I}\subseteq \Q[\vecx,y]$ be an ideal. Let $f, g, h, u, v \in \Q[\vecx,y]$ such that $f \equiv gh \bmod{\mathcal{I}}$ and $ug + vh \equiv 1 \bmod{\mathcal{I}}$. Then, we call $g', h' \in \Q[\vecx,y]$ a \emph{lift} of $g$ and $h$ if:
  \begin{enumerate}
    \item $f \equiv g' h' \bmod{\mathcal{I}^2}$,
    \item $g' \equiv g \bmod{\mathcal{I}}$ and $h' \equiv h \bmod{\mathcal{I}}$, and
    \item $\exists u', v' \in \Q[\vecx,y]$ s.t $u'g' + v'h' \equiv 1 \bmod{\mathcal{I}^2}$.\qedhere
  \end{enumerate}
\end{definition}
For the rest of the section, we define $\Ideal$ to be the ideal $\inangle{x_1, \dots, x_n}$ and $\Ideal_k := \Ideal^{2^k}$.

\begin{lemma}[Iterated monic Hensel lifting (Lemma 3.4 \cite{KSS15})]\label{lem: iterated hensel lifts}
Suppose we're given $f \in \Q[\vecx,y]$ such that $f = gh$, $g$ is monic in $y$ and $\gcd(g,h)=1$. We are also given $g_0, h_0, u_0, v_0 \in \Q[\vecx,y]$ such that $g_0 \equiv g \bmod \Ideal$, $h_0 \equiv h \bmod \Ideal$ and $u_0 g_0 + v_0 h_0 \equiv 1 \bmod \Ideal$.
Then, for all $k\in \N, k\geq 1$, there exist $g_k, h_k, u_k, v_k \in \Q[\vecx,y]$ , with each $g_k$ being monic, such that the following conditions hold:
\begin{enumerate}
  \item The pair $g_{k},h_{k}$ is a lift of $g_{k-1}, h_{k-1}$, with $u_k g_k + v_k h_k \equiv 1 \bmod \Ideal_k$; in particular, $f \equiv g_k h_k \bmod{\mathcal{I}_k}$
  \item $g_k \equiv g \bmod{\mathcal{I}_k}$ and $h_k \equiv h \bmod{\mathcal{I}_k}$ 
\end{enumerate}
Moreover, for each $k$, $g_k$ and $h_k$ are unique polynomials modulo $\mathcal{I}_k$ satisfying the above conditions when the $g_k$s are monic. For each $k$, we will call $g_k, h_k$ the \emph{$k$-th iterated lift of $g_0$, $h_0$}.
\end{lemma}

If $\deg_\vecx(g) = d$, we can choose an integer $k^*$ such that $d < 2^{k^*} \leq 2d$ and use the above Lemma to get $g_{k^*} \equiv g \bmod{\mathcal{I}_{k^*}}$, which means we can truncate $g_{k^*}$ to degree $d$ and retrieve $g$. The next lemma tells us that this can be done with reasonable bounds on the parameters of the underlying circuits.

\begin{lemma}[Small circuit for Hensel lifting (Lemma 3.6 \cite{KSS15})]\label{lem: hensel lift time}
  Let $f$ be a degree $D$ polynomial in $\Q[\vecx,y]$, computable by a $(\SP)^{(k)}$ formula of size and bit-complexity $s$, with a factorization $f = gh$ such that $\gcd(g,h) = 1$ and $g$ is monic. Let $g_0 = g \bmod{\mathcal{I}}$ and $h_0 = h \bmod{\mathcal{I}}$ be univariates in $\Q[y]$ with $\gcd(g_0,h_0) = 1$.

  Then, there are formulas $C_g, C_h$ of size and bit complexity $(sDk)^{O(k\log D)}$ that compute the $k^\text{th}$ iterated lift $g_k$,$h_k$ of $g_0$,$h_0$, where $g_k$ is monic.  More generally, if the total degree of $g_k$ is at most $d$, then the size and bit complexity of the formula for $g_k$ is at most $(sDk)^{O(\log d)}$. 
  
  Moreover, there is a deterministic algorithm, that when given the formulas for $f$ and $g_0, h_0$ and integer $k$ as input, outputs the formulas for $g_k$ and $h_k$ in time $(sDk)^{O(k\log D)}$ ( resp. $(sDk)^{O(\log d)}$ if $g_k$ has total degree $d$). 
\end{lemma}

\noindent (Proof sketch deferred to \cref{sec:appendix:deferred_proofs}.)

\subsection{Results on polynomial factorization}
We rely on the following two fundamental results on polynomial factorization for our results. The first theorem is a classical algorithm of Lenstra, Lenstra and \Lovasz{} for factoring univariate polynomials over the field of rational numbers. 
  \begin{theorem}[Factorizing polynomials with rational coefficients \cite{LLL82, GG13}]\label{thm: LLL univariate factorization}
    Let $f \in \Q[x]$ be a monic polynomial of degree $d$. Then there is a deterministic algorithm computing all the irreducible factors of $f$ that runs in time $\poly(d, t)$, where $t$ is the maximum bit-complexity of the coefficients of $f$. 
  \end{theorem}
The second result we need is an easy consequence of the results of Kopparty, Saraf and Shpilka \cite{KSS15}. They showed that an efficient deterministic algorithm for PIT for algebraic circuits implies an efficient deterministic algorithm for polynomial factorization. The formal statement below essentially invokes this for constant degree polynomials. In this case, the PIT instances also happen to be of constant degree and hence can be easily solved in time that is polynomial in the length of the coefficient vector of these polynomials. 
\begin{theorem}[\cite{KSS15}]\label{thm:det low deg factorization}
There is a deterministic algorithm that when given as input the coefficient vector of an $n$ variate polynomial $f(\vecx) \in \Q[\vecx]$ of total degree $d$, runs in time $n^{O(d^2)}$ and decides if $f$ is irreducible or not. 
\end{theorem}

\section{Computing candidate low-degree factors of multiplicity one}
\label{sec:computing-mult-one-factors}

We first present the algorithm for computing candidate low-degree factors of multiplicity one in \cref{algo: low-deg factors of multiplicity one} below. In the next section, we use this as a subroutine in \cref{algo: low-deg factors} to compute factors of all multiplicity and also eliminate those candidates that were not actual factors. 

{\small
\begin{algorithm}[H]
  \onehalfspacing
  \caption{Computing candidate degree $d$ factors of factor-multiplicity one}
  \label{algo: low-deg factors of multiplicity one}
  \SetKwInOut{Input}{Input}\SetKwInOut{Output}{Output}

  \Input{A $(\SP)^{(k)}$-formula of size $s$, bit-complexity $t$, degree $D$ computing a polynomial $f(\vecx)$.}
  \Output{A list of polynomials of degree at most $d$, that include all factors of $f$ with degree at most $d$ and multiplicity $1$.}
  \BlankLine

  Set the output list $L = \emptyset$.

  Compute hitting-set $H_1 = \mathcal{H}(d , n)$ (as defined in \cref{defn:low-hamming-weight-hitting-set}).
  \label{alg:mult-one:hitting-set-sparse}

  Compute hitting-set $H_2$ for the class of $\Sigma \inparen{(\SP)^{(k)} \cdot (\LowDeg_d)^\ast}$-formulas that have size $s' \leq (sD)^{O(d)}$. (\cref{lem: resultant computable by ugly-model}, \cref{thm: forbes quasipoly pit for resultant})
  \label{alg:mult-one:hitting-set-resultant}

  \For{$\vecalpha, \vecbeta \in H_1$ and $\vecdelta \in H_2$}{
    \label{alg:mult-one:outer-for-loop}

    Define $F(\vecx, y) = f(\vecx + \vecalpha \cdot y + \vecbeta + \vecdelta) = f(x_1 + \alpha_1 y + \beta_1 + \delta_1, \ldots, x_n + \alpha_n y + \beta_1 + \delta_n)$

    Using interpolation on the formula for $F(\vecx,y)$ (via \cref{lem:interpolation-univariate}), compute $F(\veczero,y)$ as a sum of monomials. \label{alg:mult-one:interpolate-F-sparse}

    Factorise the polynomial $F(\veczero, y)$ into irreducible factors as
    \[
      F(\veczero, y) = \sigma \cdot F_1^{e_1} \cdots F_r^{e_r}.
    \]
    where $0 \neq \sigma \in \Q$ and each $F_r$ is monic in $y$. \label{alg:mult-one:factorize-F}

    \For{$T \subseteq [r]$ of size at most $d$}{
      \label{alg:mult-one:for-loop-T}

      Define $g_0 = \prod_{i\in T} F_i^{e_i}$ and $h_0 = \sigma \cdot \prod_{i\notin T} F_i^{e_i}$, interpretted as polynomials in $\Q[\vecx,y]$ for \cref{lem: iterated hensel lifts}\label{alg:mult-one:g0-and-h0}

      \If{$\deg(g_0) > d$}{Continue to the next choice of $T$ in the current loop.}
      \label{alg:mult-one:check-g0-deg}

      Compute polynomials $u_0,v_0$ such that $u_0 g_0 + v_0 h_0 = 1$. 
      \label{alg:mult-one:compute-a0-b0}

      Use Hensel-Lifting (\cref{lem: hensel lift time}) to lift the factorisation $F(\vecx, y) = g_0(\vecx,y) \cdot h_0(\vecx,y) \bmod{I}$, where $I = \inangle{\vecx}$, to obtain algebraic circuits for $g_\ell, h_\ell$ satisfying
      \[
      F(\vecx, y) = g_\ell(\vecx,y) \cdot h_\ell(\vecx,y) \bmod{I^{2^\ell}}
      \]
      with $g_\ell$ being monic and $d < 2^\ell < 2d$. 
      \label{alg:mult-one:hensel-lift}

      Using interpolation on the circuit for $g_\ell$ (via \cref{lem: multivariate interpolation}), compute $g_\ell$ as a sum of monomials. 
      \label{alg:mult-one:interpolate-sparse}

      Add $\tilde{g} = g_\ell(\vecx - \vecdelta - \vecbeta, 0)$ to $L$.
      \label{alg:mult-one:add_gell_to_list}
      \label{alg:mult-one:for-loops-end}
    }
  }
  \Return{$L$}
\end{algorithm}
}

\medskip

Before we discuss the proof of correctness and running time of \cref{algo: low-deg factors of multiplicity one}, we state two simple observations that we use in the analysis. We defer the proofs of these observations to the end of the section. 

\begin{observation}[Size growth under a translation of low Hamming weight]\label{obs: hamming weight shift size growth}
  Let $k > 0$ be a parameter. Let $f(\vecx)$ be an $n$-variate polynomial of degree at most $D$ with $\SPsize{k}$ at most $s$. If $\vecalpha , \vecbeta \in \mathcal{H}(d,n)$, the polynomial $\tilde{f}(\vecx,y) = f(\vecx + y\vecalpha + \vecbeta)$ has $\SPsize{k}$ at most $s \cdot D^{O(d)}$. 
\end{observation}

\begin{lemma}\label{lem: resultant computable by ugly-model}
  Let $f(\vecx)$ be an $n$-variate polynomial computed by a $(\SP)^{(k)}$ formula of size $s$, and let $g(\vecx)$ be an $n$-variate degree $d$ polynomial with $g(\veczero) \neq 0$. Let $Q(\vecx)$ be the pseudo-quotient of $f$ and $g$. Then, for any variable $y\in \vecx$, the polynomial $\Res{y}{Q}{g}$ is computable by a $\uglymodel{k}{d}$ formula of size at most $s^{O(d)}$. 
\end{lemma}

\subsection{Proof of correctness of the Algorithm~\ref{algo: low-deg factors of multiplicity one}}

\begin{lemma}[Correctness of \cref{algo: low-deg factors of multiplicity one}]\label{lem: correctness of algorithm 1}
  For every input polynomial $f$ computed by $(\SP)^{(k)}$ formulas of size $s$, bit-complexity $t$, degree $D$ and any factor $g$ of degree at most $d$ with $g \mid f$ and $g^2 \nmid f$, the polynomial $g$ is included in the output list of \cref{algo: low-deg factors of multiplicity one} on input $f$. 
\end{lemma}
\begin{proof}

  \cref{algo: low-deg factors of multiplicity one} outputs a list of candidate factors; we would like to prove that every factor of $f$ with degree $\leq d$ and factor-multiplicity one will be contained in this list. Fix any specific factor $g$ of $f$, with $\deg(g) = d' \leq d$ and factor-multiplicity one, which ensures that $\gcd(g,f/g) = 1$.
  \begin{enumerate}
    \item \textbf{Make $g$ monic and $g(\vec0) \neq 0$} 
  
    The coefficient of $y^{d'}$ in $g'(\vecx,y) := g(\vecx + y \vecalpha + \vecbeta)$ is the evaluation of $\homog_{d'}(g)$ at $\vec{\alpha}$ and the constant term of $g'(\vecx,y)$ is $g'(\veczero,0) = g(\vecbeta)$. Thus by \cref{lem: non-zeros of low degree polynomials}, there is some $\vecalpha, \vecbeta \in H_1$ such that $\homog_{d'}(g)(\vecalpha) \neq 0$ and $g(\vecbeta) \neq 0$. Fix this choice of $\vecalpha,\vecbeta$. We then have that $g'(\vecx, y)$ is monic in $y$, has $\deg_y(g') = \deg(g) = d'$, and has non-zero constant term. 
    
    \item \textbf{Bound the size of $\uglymodel{k}{d}$ formula for the resultant}
  
    With the above properties, the pseudo-quotient $h'$ of $f'(\vecx, y) := f(\vecx + y \vecalpha + \vecbeta)$ and $g'(\vecx, y)$ is well-defined and is a polynomial in $\Sigma \inparen{(\SP)^{(k)} \cdot (\LowDeg_d)^\ast}$ (by \cref{lem:complexity-pseudoquotient}) of size $\poly(s,D,d) \leq \poly(sD)$. By \cref{lem: resultant computable by ugly-model}, $\Res{y}{g'}{h'} \in \Q[\vecx]$ is a non-zero polynomial computable by $\uglymodel{k}{d}$ formulas of size $(sD)^{O(d)}$. 
  
    \item \textbf{Maintain $\gcd(g,h)=1$ condition in the univariate setting by hitting the resultant}

    Let $\deg_y(h') = r$ and $h'(\vecx,y) = h_0'(\vecx) + \cdots + h_r'(\vecx) y^r$. Since $h'$ is computable by size $(sD)^{O(d)}$ formula from $\uglymodel{k}{d}$, so is the leading term $h_r'(\vecx)$ by \cref{lem:computing-hom-components}. Therefore by \autoref{obs:ugly-model-sums-products}, the polynomial $\Gamma(\vecx) = \Res{y}{g'}{h'}\cdot h'_r(\vecx)$ is also computable by $\uglymodel{k}{d}$ formulas of size $s' = (sD)^{O(d)}$. Since $H_2$ is a hitting set for $\uglymodel{k}{d}$ formulas of size $s'$, fix a $\vecdelta \in H_2$  such that $\Gamma(\vecdelta) \neq 0$ and in particular, the conditions required in \cref{obs:resultant-under-substitution} are true (note that the leading coefficient of $g'$ is just 1 by monicness). By \cref{lem: resultant property} and \autoref{obs:resultant-under-substitution}, we have that $g'(\vecdelta, y)$ and $h'(\vecdelta,y)$ are coprime polynomials. Thus, if $g''(\vecx, y) = g'(\vecx + \vecdelta, y)$ and $h''(\vecx, y) = h'(\vecx + \vecdelta, y)$ ($h'$ being the pseudo-quotient), \cref{thm: divisibility testing to pit} implies that
    \begin{align*}
    f(\vecx + \vecalpha y + \vecbeta + \vecdelta) & = g''(\vecx,y) \cdot h''(\vecx, y)\\
      \implies f(\vecalpha y + \vecbeta + \vecdelta) & = g''(\veczero,y) \cdot h''(\veczero, y)\\
      & \quad \text{with } \gcd(g''(\veczero,y), h''(\veczero,y)) = 1.
    \end{align*}
    
    \item \textbf{Univariate factorization and Hensel Lifting}

    \cref{alg:mult-one:factorize-F} thus factorises the univariate polynomial $f(\vecalpha y + \vecbeta + \vecdelta)$ and one of the sets $T$ in \cref{alg:mult-one:for-loop-T} must correspond to $g_0(y)$ chosen in \cref{alg:mult-one:g0-and-h0} to satisfy $g_0(y) = g''(\veczero,y)$ and $h_0(y) = h''(\veczero, y)$. Thus, we have a factorisation of the form 
    \begin{align*}
      f(\vecalpha y + \vecbeta + \vecdelta) & = g''(\veczero,y) \cdot h''(\veczero,y) = g_0 \cdot h_0\\
      \implies f(\vecx + \vecalpha y + \vecbeta + \vecdelta) & = g_0 \cdot h_0 \bmod{\mathcal{I}},\quad\text{where $\mathcal{I} = \inangle{\vecx}$.}
    \end{align*}
    We are therefore set-up to apply Hensel Lifting (\cref{lem: iterated hensel lifts}) to obtain $g_\ell, h_\ell$ such that $g_\ell$ is monic in $y$ and 
    \[
      f(\vecx + \vecalpha y + \vecbeta + \vecdelta) = g_\ell(\vecx,y) \cdot h_\ell(\vecx,y) \bmod{\mathcal{I}^{2^\ell}}.
    \]
  
    From the uniqueness of Hensel Lifting (which is guaranteed by \cref{lem: iterated hensel lifts}), we must have that $g_\ell(\vecx,y) = g''(\vecx,y) = g(\vecx + \vecalpha y + \vecbeta + \vecdelta)$. Thus, for this choice of $\vecalpha, \vecbeta, \vecdelta$ and $T$, we would include $g(\vecx) = g''(\vecx - \vecbeta - \vecdelta, 0)$ in the set of candidate factors in \cref{alg:mult-one:add_gell_to_list}.
  
    Finally, since the lift also ensures that there exist $u_\ell$ and $v_\ell$ such that $u_\ell g_\ell + v_\ell h_\ell = 1 \bmod{\mathcal{I}^{2^\ell}}$, we also have that $g_\ell^2 \nmid f$. 
  \end{enumerate}
\end{proof}

\subsection{Running time analysis}

\newcommand{\UniFact}{\operatorname{UniFact}}
We now bound the time complexity of the algorithm. 
\begin{lemma}[Running time of \cref{algo: low-deg factors of multiplicity one}]\label{lem:running time of algorithm 1}
Let $\epsilon > 0, d, k \in \N$ be an arbitrary constants and let $f \in \Q[\vecx]$ be a polynomial computable by a $(\SP)^{(k)}$ formula $C$ of size $s$, degree at most $D$ and bit-complexity $t$. Then, on input $C$, \cref{algo: low-deg factors of multiplicity one} terminates in time at most $(sD)^{O_{\epsilon}(kd(sD)^{\epsilon d})}\cdot t^{O(d \log d)}$.

Moreover, if $k = 1$, i.e. $f$ has sparsity at most $s$, then  \cref{algo: low-deg factors of multiplicity one} terminates in time at most $(sDt)^{(\poly(d)\log sDt)}$.
\end{lemma}
\begin{proof}
  Let $T_k^{(1)}(s,d)$ be the time-complexity to output the hitting set $H_1$ in   \cref{alg:mult-one:hitting-set-sparse} and $T_k^{(2)}(s,D,d)$ be the time-complexity to output the hitting set $H_2$ in \cref{alg:mult-one:hitting-set-resultant}. 

  From \cref{defn:low-hamming-weight-hitting-set}, we immediately have that $T_k^{(1)}(s,d) \leq s^{O(d)}$. As for $T_k^{(2)}(s,D,d)$, in the case of $k = 1$, \cref{thm: forbes quasipoly pit for resultant} shows that $T_k^{(2)}(s,D,d) \leq (sD)^{(\poly(d) \log sD)}$. For $k$ satisfying $2 \leq k = o(\log\log\log s)$, then \cref{thm: PIT from LST} shows that $T_k^{(2)}(s,D,d) \leq (sD)^{O_{\epsilon}(kd(sD)^{\epsilon d})}$ for any constant $\epsilon > 0$. 

  Using \cref{lem:interpolation-univariate}, we get that \cref{alg:mult-one:interpolate-F-sparse} takes $\poly(s,D, t)$-time. Now, each of the coefficients of $F(\veczero, y)$ has bit-complexity at most $\poly(s, D, t)$. Thus, from \cref{thm: LLL univariate factorization}, we get that $F(\veczero, y)$ can be factorized into its irreducible factors in time at most $\poly(s, D, t)$.

  There are at most $D^d$ choices for the set $T$ in \cref{alg:mult-one:for-loop-T}. For each such choice, \crefrange{alg:mult-one:g0-and-h0}{alg:mult-one:compute-a0-b0} compute formulas of size $\poly(s, D, t)$ for $g_0$, $h_0$, $u_0$, $v_0$ in time $\poly(s, D, t)$. By \cref{lem: hensel lift time}, we have that \cref{alg:mult-one:hensel-lift} takes time $(sDt)^{O(\log d)}$ to compute a formula of the same size and bit-complexity for $g_\ell$. From \autoref{lem: multivariate interpolation}, we get that we can obtain the coefficient vector of $g_{\ell}$ in time at most $(sDt)^{O(d\log d)}$. 
  \medskip

  \noindent Therefore, the overall running time of \cref{algo: low-deg factors of multiplicity one} is at most 
  \[
    T_k^{(1)}(s,d) \cdot T_k^{(2)}(s,D,d) \cdot D^d \cdot  \poly(s, D, t) \cdot (sDt)^{O(d\log d)}  \, .
  \]
Plugging in the estimates for $T_k^{(1)}(s,d)$, $T_k^{(2)}(s,D,d)$, we get the overall bound of $(sD)^{O_{\epsilon}(kd(sD)^{\epsilon d})}\cdot t^{O(d \log d)}$ for $k > 1$, which is essentially dominated by $T_k^{(2)}(s,D,d)$. 
  
  When $f$ has sparsity $s$, then as discussed in the proof, $ T_k^{(2)}(s,d)$ is at most $(sD)^{(\poly(d) \log sD)}$. Plugging this back in the above expression, we get that the running time is at most $(sDt)^{(\poly(d)\log sDd)}$. 
\end{proof}

\subsection{Proof of structural lemmas} \label{section: proof of structural lemmas}
In this subsection, we include the proofs of \autoref{obs: hamming weight shift size growth} and \cref{lem: resultant computable by ugly-model}. This completes the analysis of \cref{algo: low-deg factors of multiplicity one}. 

\begin{proof}[Proof of \autoref{obs: hamming weight shift size growth}]
  By definition of $\mathcal{H}(d,n)$ (\cref{defn:low-hamming-weight-hitting-set}), the transformation $\vecx \mapsto \vecx + y \vecalpha + \vecbeta$ takes a monomial $\prod_{i\in [n]}{x_i^{e_i}}$ to $\inparen{\prod_{i\in T}{\inparen{x_i + \alpha_i y + \beta_i}^{e_i}}}\cdot \inparen{\prod_{i \in [n]\setminus T} x_i^{e_i}}$, for some $T \subseteq [n]$ s.t. $|T| = d$. If we expand $\prod_{i\in T}{\inparen{x_i + \alpha_i y + \beta_i}^{e_i}}$ into a sum of monomials, we will get at most $D^{O(d)}$ monomials (when $\sum_i e_i \leq D$). Expanding each $\prod_{i\in [n]}{\inparen{x_i + \alpha_i y + \beta_i}^{e_i}}$ at the bottom layer into a sum of monomials this way, we get the required $(\SP)^{(k)}$ formula with size at most $s\cdot D^{O(d)}$.
\end{proof}

\begin{proof}[Proof of \cref{lem: resultant computable by ugly-model}]
  Let $\vecx' = \vecx \setminus \inbrace{y}$ and let $\mathcal{C}$ be the class $\uglymodel{k}{d}$. Let us assume that $\deg_y(Q) = D \leq s$ and $\deg_y(g) = d$. 
  By \cref{lem:complexity-pseudoquotient}, we have that $Q(\vecx)$ is computable by a $\mathcal{C}$-formula of size at most $\poly(s,d)$. Let us consider the $(D + d) \times (D + d)$ Sylvester matrix $\Gamma$ of $Q$ and $g$ with respect to the variable $y$ whose determinant is $\Res{y}{Q}{g}$. 
  \begin{align*}
    Q(\vecx) & = Q_0(\vecx') + y \cdot Q_1 (\vecx') + \cdots + y^D \cdot Q_D(\vecx')\\
    g(\vecx) & = g_0(\vecx') + \cdots + y^d \cdot g_d(\vecx')\\
  \Gamma & = 
    \begin{bmatrix}
        Q_0  & Q_1    & \dots  &        & Q_{D} &        &       \\
            & \ddots & \ddots &        & \ddots & \ddots &       \\
            &        & Q_0    & Q_1    &        &  \dots & Q_{D} \\
      g_{0}  & \dots  &        & g_{d}  &        &        &       \\
            &g_{0}   & \dots  &        & g_{d}  &        &       \\
            &        & \ddots & \ddots &        & \ddots &       \\
            &        &        & g_{0}  & \dots &        &g_{d} 
    \end{bmatrix}
  \end{align*}
  Note that, by \cref{lem:interpolation-univariate}, each of the $Q_i$'s are computed by a $\mathcal{C}$-formula of size $\poly(s,D)$ and each $g_i$ is a polynomial of degree at most $d$. 

  \newcommand{\Top}{\operatorname{Top}}
  \newcommand{\Bot}{\operatorname{Bot}}
  For a subset $S$ of rows and $T$ of columns, we will use $\Gamma(S,T)$ to refer to the submatrix restricted to the rows in $S$ and columns in $T$, and let $\Top = \set{1,\ldots, d}$ and $\Bot = \set{d+1,\ldots, d+D}$. The determinant of $\Gamma$ can then be expressed as
  \[
    \det(\Gamma) = \Res{y}{Q}{g} = \sum_{T \in \binom{[D+d]}{d}} \det(\Gamma(\Top,T)) \cdot \det(\Gamma(\Bot, \overline{T}))
  \]
  For every choice of $T$, the polynomial $\det(\Gamma(\Top,T))$ is the determinant of a $d\times d$ matrix each of whose entries are computable by $s' = \poly(s, D)$ sized $\mathcal{C}$-formulas. Therefore, using \autoref{obs:ugly-model-sums-products}, the polynomial $\det(\Gamma(\Top,T))$ is computable by $\mathcal{C}$-formulas of size at most $(sD)^{O(d)}$. 

  The polynomial $\det(\Gamma(\Bot, \overline{T}))$ is a degree $D$ polynomial combination of $g_0,\ldots, g_d$ and can therefore be expressed as 
  \begin{align*}
    \det(\Gamma(\Bot, \overline{T})) & = \sum_{i=1}^{D^{d+1}} a_i \cdot g_0^{e_{i,0}} \cdots g_d^{e_{i,d}}\\
    & = \sum_{i=1}^{D^{d+1}} a_i \cdot \inparen{\sum_{j=1}^{D^{O(d)}} b_{ij} \cdot f_{ij}^{e_{ij}}} \quad \text{(using \cref{lem:fischers-trick})}.
  \end{align*}
  for some polynomials $f_j$ of degree at most $d$. Thus, using \autoref{obs:ugly-model-sums-products} again, we have that $\Res{y}{Q}{g}$ is computable by $\uglymodel{k}{d}$ formulas of size at most $D^{O(d)} \cdot (sD)^{O(d)} = (sD)^{O(d)}$. 
\end{proof}

\section{Computing factors of all multiplicity}\label{section: arbitrary multiplicity}

The following lemma essentially shows that the multiplicity of any factor $g$ of a given polynomial $f$ can be reduced by working with appropriate partial derivatives of $f$, with respect to variables that are present in $g$. This naturally yields an algorithm that uses \cref{algo: low-deg factors of multiplicity one} as a subroutine, and computes all irreducible factors of $f$.

\begin{lemma}[Reducing factor multiplicity]\label{lem:factor-multiplicity-reduction}
  Let $f(\vecx), g(\vecx) \in \Q[\vecx]$ be non-zero polynomials and let $x \in \vecx$ be such that $\partial_x(g) \neq 0$ and $g$ is square-free. Then, the factor-multiplicity of $g$ in $f$ (i.e. the integer $a$ satisfying $g^a \mid f$ and $g^{a+1} \nmid f$) is also the smallest non-negative integer $a$ such that $g \nmid \frac{\partial^a f}{\partial x^a}$. 
\end{lemma}
\begin{proof}
  If the factor-multiplicity of $g$ in $f$ is zero, i.e. $g \nmid f$, then claim is clearly true. Thus let us assume that the factor-multiplicity of $g$ in $f$ is $a \geq 1$. It suffices to show that the factor-multiplicity of $g$ in $\partial_x(f)$ is exactly $a-1$. 

  Suppose $f = g^a \cdot h$ where $\gcd(g,h) = 1$. Then, 
  \[
    \partial_x f = \partial_x(g^a) \cdot h + g^a \cdot \partial_x(h) = g^{a-1} \cdot (a \cdot \partial_x(g) \cdot h + g \cdot \partial_x(h)).
  \]
  Hence, we have that the factor-multiplicity of $g$ in $\partial_x(f)$ is at least $(a-1)$. 

  On the other hand, we have that $\partial_x(g) \neq 0$ and $g$ is square-free and hence $\gcd(g, \partial_x(g)) = 1$. Therefore
  \[
    \gcd(g, a\cdot g \cdot \partial_x(h) + h \cdot \partial_x(g)) = \gcd(g, h \cdot \partial_x(g)) = \gcd(g, h) = 1
  \]
  and hence $g^a \nmid \partial_x (f)$ and therefore the factor-multiplicity of $f$ 
\end{proof}

\noindent We are now ready to describe the algorithm.

{\small 
\begin{algorithm}[H]
  \caption{Computing list of all degree $d$ irreducible factors and their multiplicities}
  \label{algo: low-deg factors}
  \SetKwInOut{Input}{Input}\SetKwInOut{Output}{Output}

  \Input{A $(\SP)^{(k)}$-formula of size $s$, bit-complexity $t$, degree $D$ computing a polynomial $f(\vecx)$.}
  \Output{A list of all irreducible factors $f$ of degree at most $d$ and their multiplicities.}
  \BlankLine

  Set the output list $L = \emptyset$.
  
  Set the intermediate candidates list $L' = \emptyset$.

  Compute hitting-set $H_1 = \mathcal{H}(d , n)$ (as defined in \cref{defn:low-hamming-weight-hitting-set}).

  \For{$\vecalpha \in H_1$}{
    \label{alg:all-mult:for-alpha}
    Define $F(\vecx, y) = f(\vecx + \vecalpha \cdot y) = f(x_1 + \alpha_1 y, \ldots, x_n + \alpha_n y)$

    \For{$i = 0,1,\ldots, \deg(F)$}{
      \label{alg:all-mult:for-i computing intermediate candidates list}

      Define $\tilde{F}(\vecx,y) = \frac{\partial^i F}{\partial y^i}$.

      Compute the list $\tilde{L}$ of all candidate degree $d$ multiplicity-one factors of $\tilde{F}(\vecx,y)$ using \cref{algo: low-deg factors of multiplicity one}. 
      \label{alg:all-mult:run-alg1}

      \ForEach{$\tilde{g}(\vecx,y) \in \tilde{L}$}{
        \label{alg:all-mult:foreach-Ltilde}

        Add $g(\vecx) := \tilde{g}(\vecx,0)$ to $L'$.
        \label{alg:all-mult:unshift-and-add-to-candidates}

        \label{alg:all-mult:foreach-Ltilde-end}
        }
    }
  }

  \For{$g \in L'$}{
    \label{alg:all-mult:for-factor-prune}

    \textbf{if} $g$ is not irreducible \textbf{then} skip to the next iteration. \label{alg:all-mult:irred-test}

    Let $x$ be a variable that $g$ depends on, so that $\partial_x(g) \neq 0$. 

    Find the smallest non-negative integer $e$ such that $g \nmid \frac{\partial^e f}{\partial x^e}$. 
    
    \textbf{if} $e > 1$ \textbf{then} add $(g,e)$ to the list $L$.
    \label{alg:all-mult:for-factor-prune-end}
  }
  \Return{$L$}
\end{algorithm}
}

\begin{lemma}[Correctness of \cref{algo: low-deg factors}]
  \label{lem:alg2_correctness}
For every input polynomial $f$ computed by a $(\SP)^{(k)}$ formula of size $s$, degree $D$, bit-complexity $t$ and $d \in \N$, the list $L$ output by \cref{algo: low-deg factors} is precisely the list of all irreducible factors of $f$  of degree at most $d$ (up to scalar multiplication) along with their multiplicities in $f$.   
\end{lemma}
\begin{proof}
  From \crefrange{alg:all-mult:for-factor-prune}{alg:all-mult:for-factor-prune-end} and \cref{lem:factor-multiplicity-reduction}, it is clear that any $(g,e)$ in the output list ensures that $g$ is an irreducible polynomial, $g^e \mid f$ and $g^{e+1} \nmid f$. Thus, it suffices to show that for every irreducible polynomial $g$ such that $\deg(g) \leq d$ and $g \mid f$, some non-zero scalar multiple of $g$ is under consideration in the list $L'$. Fix any such irreducible factor $g$ of degree at most $r \leq d$ and let its factor-multiplicity be $e$

  By \cref{lem: non-zeros of low degree polynomials}, there is some $\vecalpha \in H_1$ such that $\homog_r(g)(\vecalpha) \neq 0$, where $r$ is the total degree of $g$. Thus, for this choice of $\vecalpha$, we have that $g'(\vecx,y) = g(\vecx + y \vecalpha)$ is a factor of $F(\vecx,y) = f(\vecx + y \vecalpha)$ and $g'$ is monic in $y$ and has factor-multiplicity $e$. By \cref{lem:factor-multiplicity-reduction}, we have that $g'$ has factor-multiplicity one in $\tilde{F}(\vecx,y) := \frac{\partial^{e-1} F}{\partial y^{e-1}}$. Thus, by the correctness of \cref{algo: low-deg factors of multiplicity one} (\cref{lem: correctness of algorithm 1}), a non-zero multiple of the polynomial $g'(\vecx,y)$ must be included in the list $\tilde{L}$ in \cref{alg:all-mult:run-alg1}. Therefore, a non-zero multiple of $g(\vecx) = g'(\vecx,0)$ will be added to $L'$ in \cref{alg:all-mult:unshift-and-add-to-candidates}. 
\end{proof}

\begin{lemma}[Running time of \cref{algo: low-deg factors}]
  \label{lem:alg2_runningtime}
Let $\epsilon > 0, k, d \in \N$ be  arbitrary constants. Let $f \in \Q[\vecx]$ be a polynomial computable by a $(\SP)^{(k)}$ formula $C$ of size $s$, degree at most $D$ and bit-complexity $t$. Then, on input $C$ and $d \in \N$, \cref{algo: low-deg factors of multiplicity one} terminates in time at most $(sDt)^{O( k d(sDt)^{\epsilon d})}$.

Moreover, if $k = 1$, i.e. $f$ has sparsity at most $s$, then  \cref{algo: low-deg factors of multiplicity one} terminates in time at most  $(snDt)^{O(\poly(d) \cdot \log snDt)}$.
\end{lemma}
\begin{proof}
From \cref{defn:low-hamming-weight-hitting-set}, we have the size of the set $H_1$ is $n^{O(d)}$. The time complexity of computing a formula for $F$ from the given formula for $f$ is at most $O(sD)$. From \cref{lem: partial derivatives with respect to one variable}, we have that $(\SP)^{k+1}$ formulas for all the $y$ derivatives of $F$ can be computed in time at most $\poly(s, D, t)$, which is also a bound on the bit-complexity and the size of these formulas. \cref{algo: low-deg factors of multiplicity one} is invoked at most $D$ times. 

The total time taken to construct the list $L'$ is at most $D\cdot T_1$, where $T_1$ is the time taken by \cref{algo: low-deg factors of multiplicity one} on inputs with formula size and bit-complexity $\poly(s, D, t)$, and degree parameter $d$. $D \cdot T_1$ is also an upper bound on the size of the list of candidate factors $L'$.

Now, for each $g \in L'$, from \cref{thm:det low deg factorization}, we have that the irreducibility test in \cref{alg:all-mult:irred-test} takes at most $(sDt)^{O(d^2)}$ time. There are at most $D$ instances of divisibility test performed to determine the exact multiplicity in $f$ of each $g \in L'$. This requires computing the corresponding derivatives, which as discussed in the previous paragraph, takes time $\poly(s, D, t)$ and outputs a formula of size and bit-complexity $\poly(s, D, t)$ for the derivatives, and then doing a divisibility test, the time complexity of which we denote by $T_2$. 

Therefore, the total time taken by the algorithm is at most $(n^{O(d)} \cdot \poly(s, D, t) \cdot D \cdot T_1) +  (D\cdot T_1 \cdot (sDt)^{O(d^2)}  \cdot \poly(s, D, t) \cdot T_2)$. 

Now, if $f$ is $s$ sparse, i.e. $k =1$, then from \autoref{defn:low-hamming-weight-hitting-set}, we have that every vector in $H_1$ has at most $d$ non-zero coordinates. Thus, from \autoref{obs: hamming weight shift size growth}, for every $\vecalpha \in H_1$, $F(\vecx, y) = f(\vecx + \vecalpha \cdot y)$ has sparsity and bit-complexity at most $s' \leq s\cdot D^d$. Note that the derivatives of arbitrary order of $F$ with respect to any variable also have the same bound on their sparsity and bit-complexity of coefficients. Thus, in this case, from \autoref{lem:running time of algorithm 1}, $T_1 \leq (sDt)^{\poly(d)\log sDt}$. From \autoref{thm: det algo for low degree dividing sparse}, we have that $T_2 \leq (snD)^{O(d\log^2 snD)}$. Therefore, the overall running time of the algorithm is at most $(snDt)^{O(\poly(d) \cdot \log snDt)}$.

On the other hand, if $k > 1$, then from \autoref{lem:running time of algorithm 1}, $T_1 \leq (sD)^{O_{\epsilon}(kd(sD)^{\epsilon d})}\cdot t^{O(d \log d)}$. To bound $T_2$ in this case, we note from \autoref{cor: divisibility to pit constant depth}, this divisibility testing instances reduce to PIT instances for $(\SP)^{(k+1)}$ formula of size and bit-complexity at most $\poly(s, D, t)$ and from \autoref{thm: PIT from LST}, this can be done in at most $(sDt)^{O(k(sDt)^{\epsilon})}$ time for the arbitrary constant $\epsilon$ chosen in the beginning. Thus, the total time taken is at most $(sDt)^{O_{\epsilon}(kd(sDt)^{\epsilon d})}$.
\end{proof}

\noindent
\cref{lem:alg2_correctness} and \cref{lem:alg2_runningtime} together imply our main theorems \cref{thm: main low depth ckts} and \cref{thm: main sparse}.

\section{Open problems}\label{section: open problems}

We conclude with some open problems. 

\begin{itemize}
\item Perhaps the most natural open problem here is to obtain efficient deterministic algorithms that completely factor sparse polynomials or more generally, polynomials with constant depth formulas (and not just obtain low degree factors). In the absence of better structural guarantees for the factors (for instance, if they are sparse or have small constant depth formulas), we can seek algorithms that output general algebraic circuits for these factors.  
\item Obtaining improved structural guarantees on the factors of polynomials that are sparse or have small constant depth formulas as mentioned in the first open problem is another very interesting open problem. 
\item A first step towards obtaining deterministic algorithms for general factorization of polynomials with small constant depth formulas could be to  design deterministic algorithms for computing \emph{simple} factors of such polynomials. While the notion of simplicity discussed in this paper is that of low degree factors, there are other natural notions that seem very interesting. For instance, can we design an efficient deterministic algorithm that outputs all the sparse irreducible factors of a constant depth formula ?
\item As alluded to in the introduction, polynomial factorization algorithms have found numerous applications in computer science. It would be interesting to understand if there are applications of  deterministic factorization algorithms in general, and in particular the algorithms for computing low degree factors described in this paper. 

\end{itemize}

\ifblind
\else
\section*{Acknowledgements}
A part of this work was done while the first two authors were at the Workshop on Algebraic Complexity organised at the University of Warwick in March 2023 by Christian Ikenmeyer. We thank Christian for the invitation and the delightful and stimulating atmosphere at the workshop.  
\fi

\bibliographystyle{customurlbst/alphaurlpp}
{\let\thefootnote\relax
\footnotetext{\textcolor{\gitinfonotecolour}{\gitinfonote \easteregg}
}}
\bibliography{crossref,references}

\appendix

\section{Deferred proofs}
\label{sec:appendix:deferred_proofs}

\subsubsection*{Circuit/formula bit-complexity}

\begin{proof}[Proof of \cref{lem:evaluation-bit-complexity}]
  We will prove an equivalent statement: the numerator and denominator of $f(\veca)$ have absolute value at most $2^{s\cdot b}$. We prove this by induction on the size of the formula. We will use $N(\cdot)$ and $D(\cdot)$ to denote the numerator and denominator of some rational number. 
  
  Base case: when $\size(C) = 1 = s$, there is a single leaf node in the formula that reads and outputs a single rational number of bit-complexity $b$, thus $\bit(f(\veca)) = b \leq s \cdot b $.
  The induction hypothesis is that for all formulas $C$ with $\size(C)\leq S$ (for some $S\geq 1$), $\bit(f(\veca)) \leq \bit(C)\cdot b$. For the induction step, we look at formulas $C$ with $\size(C) = S+1$, and we consider two cases:
  \begin{enumerate}
    \item When the top gate is a sum gate: $f = \sum_{i=1}^k{\alpha_i g_i(\vecx)}$, with $C_i$ being the formula computing $g_i$ and $\alpha_i$s being scalars from $\Q$.
    \begin{flalign*}
      f(\veca) &= \sum_{i=1}^k{\alpha_i g_i(\veca)}  && \\
      \abs{D(f(\veca))} &=\abs{\prod_{i=1}^k{D(\alpha_i)D(g_i(\veca))}} && \\
      &\leq \prod_{i=1}^k{2^{\bit(\alpha_i)}2^{\bit(g_{i}(\veca))}} && \\
      &\leq \prod_{i=1}^k{2^{\bit(\alpha_i)}2^{\bit(C_{i})\cdot b}} && \text{(induction hypothesis)} \\
      &= 2^{\sum_{i=1}^k{\inparen{\bit(\alpha_i)+\bit(C_{i})\cdot b}}} \leq 2^{\bit(C)\cdot b} &&
   \end{flalign*} 
    \begin{flalign*}
      \abs{N(f(\veca))}&\leq \sum_{i=1}^k{\abs{N(\alpha_i)}\abs{N(g_i(\veca))}\prod_{j \neq i}{\abs{D(\alpha_j)}\abs{D(g_j(\veca))}}} && \\
      &\leq \sum_{i=1}^k{2^{\bit(\alpha_i) + \bit(g_i(\veca))}2^{\sum_{j \neq i}{\bit(\alpha_j)+\bit(g_j(\veca))}}} && \\
      &\leq \sum_{i=1}^k{2^{\bit(\alpha_i) + \bit(C_i) \cdot b}2^{\sum_{j \neq i}{\bit(\alpha_j)+\bit(C_j)\cdot b}}} && \text{(induction hypothesis)}\\
      &\leq \sum_{i=1}^k{2^{\sum_{j=1}^k{\bit(\alpha_j)+\bit(C_j) \cdot b}}} \leq 2^{k+ \sum_{j=1}^k{\bit(\alpha_j)+\bit(C_j) \cdot b}} \leq 2^{\bit(C)\cdot b} && 
    \end{flalign*}
    Thus, $\bit(f(\veca)) = \max\{\bit(N(f(\veca)),D(f(\veca)))\} \leq \bit(C)\cdot b$.
  
    \item When the top gate is a product gate: $f = \prod_{i=1}^k{\alpha_i g_i(\vecx)}$. The proof for the denominator in the case of sum gate will work here for both the numerator and the denominator. The required bound follows.
  \end{enumerate}
  
\end{proof}

\subsubsection*{Relevant subclasses of algebraic circuits}

\begin{proof}[Proof of \autoref{obs:ugly-model-sums-products}]
  We prove the size upper bounds here; the bit-complexity upper bounds proceed along exactly the same lines.
  The size upper bound for the sum is immediate and hence we only need to focus on the product. Let the expression for each $P_r$ be
  \begin{align*}
    P_r & = \sum_i P_{r,i}^{(r)} \cdot g_{r,i}^{a_{r,i}}\\
    \implies \prod P_r & = \sum_{r_1,\ldots, r_t} \inparen{P_{1,r_1} \cdots P_{t,r_t}} \cdot \inparen{g_{1,r_1}^{a_{1,r_1}} \cdots g_{t,r_t}^{a_{t,r_t}}}
  \end{align*}
  where each $P_{i,j}$ is computed by $(\SP)^{(k)}$ formulas of size at most $s$, and each $g_{i,j}$ is a polynomial of degree at most $d$. 

  Each $\inparen{P_{1,r_1} \cdots P_{t,r_t}}$ is computed by a $(\SP)^{(k)}$ formula of size at most $s^t$. By \cref{lem:fischers-trick}, $g_{1,r_1}^{a_{1,r_1}} \cdots g_{t,r_t}^{a_{t,r_t}}$ can be expressed as a sum $\sum_{\ell = 1}^{s^t} f_\ell^{D}$ where $D = \sum_j a_{j,r_j}$ and each $f_\ell$ is a degree polynomial of degree at most $d$. Thus, $\prod_r P_r$ is computable by a $\uglymodel{k}{d}$ formula of size at most $s^{O(t)}$. 
\end{proof}

\subsubsection*{Polynomial identity testing}

\begin{proof}[Proof of \cref{lem: non-zeros of low degree polynomials}]
  Since $f$ is a non-zero polynomial of degree at most $d$, there is a monomial $\vecx^{\vece}$ of degree at most $d$ with a non-zero coefficient in $f$. Let $S$ be the support of the monomial $\vecx^{\vece}$, i.e., $S = \{x_i : e_i \neq 0\}$. Clearly, $|S| \leq d$. We now consider the polynomial $\tilde{f}$ obtained from $f$ by setting all the variables $x_j$ not in the set $S$ to zero. Since $f$ has a non-zero monomial with support contained in the set $S$, $\tilde{f}$ continues to be a non-zero polynomial of degree at most $d$. Moreover, it is a $d$ variate polynomial since it only depends on the variables in $S$. From \cref{lem: SZ lemma}, we get that for any subset $T_d$ of $\Q$ of cardinality at least $d+1$, there exists a vector $\vecb \in {T_d}^{d}$ such that $\tilde{f}(\vecb) \neq 0$. Let $\veca \in \Q^n$ to be such that for every $i \in S$, $a_i = b_i$ and for every $i \notin S$,  $a_i = 0$. Then, $f(\veca) = \tilde{f}(\vecb) \neq 0$. Moreover, $\veca $ is in $\mathcal{H}(d, n)$.
\end{proof}

\begin{proof}[Proof of \cref{lem: multivariate interpolation}]
  Let $T_d$ be the set $\{0, 1, 2, 3, \ldots, d\}$ and let $\mathcal{H}(d, n)$ be the set of points defined in \cref{defn:low-hamming-weight-hitting-set}, i.e.,
   \[
      \mathcal{H}(d, n) = \setdef{(a_1,\ldots, a_n)}{S \in \binom{[n]}{\leq d} \;,\; a_i \in T_d \text{ for all } i\in S \text{ and } a_j = 0 \text{ for all } j\notin S}.
    \]
  From \cref{lem: non-zeros of low degree polynomials}, we know that every non-zero polynomial $f$ of degree at most $d$ must evaluate to zero on some point of $\mathcal{H}(d, n)$. In other words, two distinct degree $d$ polynomials $f$ and $g$ cannot agree on every point of $\mathcal{H}(d, n)$. An immediate consequence of this is that if we are given the evaluations of an unknown polynomial $f$ on all points of $\mathcal{H}(d, n)$, and we view each of these evaluations as a linear constraint on the unknown coefficients of $f$, then this linear system has a unique solution. 
  
  Based on this observation, a natural algorithm for computing the coefficient vector of $C$ is the following, we evaluate the given formula on every input in $\mathcal{H}(d, n)$, set up the linear system on the coefficients of $C$ obtained from these evaluations, and use any standard linear system solver over $\Q$ to solve this system. 
  
  Note that the size of this linear system is at most $n^{O(d)}$, and from \autoref{lem:evaluation-bit-complexity}, of the constants in this linear system is at most $\poly(s, b, d)$. Thus, this linear system can be solved in time $\poly(s, b, d, n^d) \leq \poly(s, b, n^d)$ time as claimed.
\end{proof}
  
\subsubsection*{Deterministic divisibility testing and PIT}

\begin{proof}[Proof of \cref{cor: divisibility to pit constant depth}]
  The proof essentially follows immediately from \cref{thm: divisibility testing to pit}. From \cref{thm: divisibility testing to pit}, we have that $g$ divides $f$ if and only if $R(\vecx) := f(\vecx) - g(\vecx) Q(\vecx) \equiv 0$, where $Q$ is the pseudo-quotient of $f$ and $g$. It suffices to show that $R(\vecx)$ has $\mathcal{C} = \Sigma \inparen{(\SP)^{(k)} \cdot (\LowDeg_d)^\ast}$ formulas of size $\poly(s,d)$, and since $f(\vecx) \in (\SP)^{(k)}$, it suffices to bound the size of $\mathcal{C}$-formulas computing $g(\vecx) \cdot Q(\vecx)$. 
  
  By \cref{lem:complexity-pseudoquotient}, the pseudo-quotient $Q(\vecx)$ is computable by $\mathcal{C}$-formulas of size $\poly(s,d)$. Let one such computation be of the form
    \begin{align*}
    Q(\vecx) & = \sum_{i} f_i \cdot g_i^{e_i}\quad\text{where each $f_i \in (\SP)^{(k)}$ and $\deg(g_i) \leq d$ and each $e_i \leq s$}\\
    \implies g(\vecx) Q(\vecx) & = \sum_i f_i \cdot (g \cdot g_i^{e_i})
    \end{align*}

  From \cref{lem:fischers-trick}, note that any term of the form $(g \cdot h^{e})$ can be expressed as
  \[
  g \cdot h^e = \sum_{i=1}^{\poly(e)} \beta_i \cdot (g + \alpha_i h)^{e+1}
  \]
  for field constants $\alpha_i$'s and $\beta_i$'s. Thus, feeding this in the above expression for $g \cdot Q$, we have 
  \[
  g(\vecx) \cdot Q(\vecx)  = \sum_i\sum_j f_i \cdot\tilde{g}_{ij}^{e_{ij}}
  \]
  for polynomial $\tilde{g}_{ij}$ of degree at most $d$, and thus is also a $\mathcal{C}$-formula of size at most $\poly(s,d)$. Therefore, $R(\vecx) = f(\vecx) - g(\vecx) Q(\vecx)$ is also computable by $\mathcal{C}$-formulas of size $s' = \poly(s,d)$. Thus, we can check if $g$ divides $f$ by checking if $R(\vecx)\equiv 0$ (by \cref{thm: divisibility testing to pit}) which can be done in $T(k,d,s')$ time as claimed. 
\end{proof}

\subsubsection*{Hensel lifting}

\begin{proof}[Proof sketch of \cref{lem: hensel lift time}]
  
  As indicated earlier, the lemma is almost an immediate consequence of Lemma 3.6 in \cite{KSS15}. The precise statement there gives a circuit $\tilde{C_k}$ of size and bit-complexity $\poly(s, D, 2^k)$ for $g_k, h_k$. We notice that without loss of generality, the degree of $g_k, h_k$ and hence of  $\tilde{C_k}$ can be assumed to be at most $(D + 2^k)$ since the $y$ degree is at most $D$ and the $x$ degree is at most $2^k$. This incurs at most a polynomial blow up in the circuit size. 
  
  Now, to go from circuits for $g_k, h_k$ to formulas computing these polynomials, we just invoke the classic depth reduction result of Valiant, Skyum, Berkowitz and Rackoff \cite{VSBR83}, which states that given an $n$-variate degree-$\Delta$ polynomial $f$ with an arithmetic circuit $\Phi$ of size $s$, there is an arithmetic circuit $\Phi'$ that computes $f$, has size $\poly(s,n,\Delta)$ and depth $O(\log \Delta)$.

  Thus we have a formula of size (and bit-complexity) at most $\poly(s, D, 2^k)^{\log (D + 2^k)} \leq (sDk)^{k \log D}$.
  Note that a better bound of $d$ on the total degree of $g_k$ implies that the size and bit-complexity of the formula for $g_k$ is at most $(sDk)^{O(\log d)}$. 
  \end{proof}

\end{document}